\documentclass[conference]{IEEEtran}
\IEEEoverridecommandlockouts
\usepackage{mathptm}
\usepackage{amsmath}
\usepackage{amssymb}
\usepackage{graphicx}
\usepackage{times}
\usepackage{epsfig}
\usepackage{subfigure}
\usepackage{cite}
\usepackage{graphicx}
\usepackage[utf8]{inputenc}
\usepackage{amsthm}
\usepackage{algorithm}
\usepackage{algorithmic}
\usepackage{xcolor}
\usepackage{booktabs}
\usepackage{multirow}

\newtheorem{definition}{\bf Definition}
\newtheorem{lemma}{\bf Lemma}
\newtheorem{theorem}{\bf Theorem}

\newcommand{\ra}[1]{\renewcommand{\arraystretch}{#1}}

\def\BibTeX{{\rm B\kern-.05em{\sc i\kern-.025em b}\kern-.08em
    T\kern-.1667em\lower.7ex\hbox{E}\kern-.125emX}}
\begin{document}

\title{Optimal $k$-Coverage Charging Problem}
\author{Xuan Li, Miao Jin\\
Center for Advanced Computer Studies, University of Louisiana at Lafayette, Lafayette, LA 70504}

\maketitle

\begin{abstract}
	
	Wireless rechargeable sensor networks, consisting of sensor nodes with rechargeable batteries and mobile chargers to replenish their batteries, have gradually become a promising solution to the bottleneck of energy limitation that hinders the wide deployment of wireless sensor networks (WSN). In this paper, we focus on the mobile charger scheduling and path optimization scenario in which the $k$-coverage ability of a network system needs to be maintained. We formulate the optimal $k$-coverage charging problem of finding a feasible path for a mobile charger to charge a set of sensor nodes within their estimated charging deadlines under the constraint of maintaining the $k$-coverage ability of the network system, with an objective of minimizing the energy consumption on traveling per tour. We show the hardness of the problem that even finding a feasible path for the trivial case of the problem is an NP-complete one.
	
	We model the problem and apply dynamic programming to design an algorithm that finds an exact solution to the optimal $k$-coverage charging problem. However, the computational complexity is still prohibitive for large size networks. We then introduce Deep Q-learning, a reinforcement learning algorithm to tackle the problem. Traditional heuristic or Mixed-Integer and Constraint Programming approaches to such combinatorial optimization problems need to identify domain-specific heuristics, while reinforcement learning algorithms could discover domain-specific control policies automatically. Specifically, Deep Q-Learning applies Deep Neural Network (DNN) to approximate the reward function that provides feedback to the control policy. It can handle complicated problems with huge size of states like the optimal $k$-coverage charging problem. It is also fast to learn the reward function. We implement three other heuristic algorithms for comparison and then conduct extensive simulations with experimental data included. Results demonstrate that the proposed Deep Q-Learning algorithm consistently produces optimal solutions or closed ones and significantly outperforms other methods.
	
\end{abstract}


\section{Introduction}\label{sec:intro}

Wireless rechargeable sensor networks, consisting of sensor nodes with rechargeable batteries and mobile chargers to replenish their batteries, have gradually become a promising solution to the bottleneck of energy limitation that hinders the wide deployment of wireless sensor networks (WSN)~\cite{shi2011renewable,Rechargeable2015,wirelessrechargeable2016,liang2017approximation}. The mobile charger scheduling and path optimization problem optimizes the trajectory of a mobile charger (e.g., a mobile robot) to maintain the operation of a WSN system. Variants are studied by considering different optimization goals, application scenarios, and constraints.  Some research works extend the mobile charger scheduling and path optimization problem from one charger to multiple ones~\cite{Liang:2016:MLR:2925994.2898357,Lin18-INFOCOM} and static sensor nodes to mobile ones~\cite{he2014mobile,chen2016charge}.

%{\color{red}Inductive Coupling (IC), Electromagnetic Radiation (ER) and Magnetic Resonant Coupling (MRC)~\cite{xie2013modeling} are three current technologies in Wireless Energy Transfer (WET) area. IC supported for short charging distance and need accurate alignment for charging direction. It is applied to electric toothbrush and wireless charging for mobile devices. ER consists of omnidirectional ER and unidirectional ER. Omnidirectional ER is used on low power WSNs because its energy efficiency drop of rapidly over distance. Unidirectional ER supports kilometer range charging and it requires Line of Sight (LOS), large scale of devices and complicated tracking methods, so it is used in unmanned plane. MRC is a promising technology and currently widely used in charging Electrical Vehicular (EV)\cite{Qualcomm,WiTricity}. MRC can reach over 90 percent energy efficiency in $10-25$cm range \cite{WiTricity} and it is insensitive to environment change. MRC technique is a promising technique applied to mobile charger which used to charge the WSNs.}

Many of the prior works focused on the problem of maximizing the number of nodes charged within a fixed time horizon or energy constraint with the assumption that each sensor node contributes equally to the sensing quality of a network. However, a full area coverage is a basic requirement of WSN deployment to monitor a certain area. Multiple coverage, where each point of the field of interest (FoI) is covered by at least
$k$ different sensors with $k>1$ ($k$-coverage) , is often applied to increase the sensing accuracy of data fusion and enhance the fault tolerance in case of node failures~\cite{Yang06onconnected,Simon07dependablek,Handbook2008,VariableRadii2009,Multiple-Coverage2011,li2015autonomous}. Existing approaches to achieve $k$-coverage deploy a set of sensor nodes over a FoI either in a randomized way~\cite{Kumar:2004:KMS:1023720.1023735,Hefeeda:2007:RKA:2931311.2931587} or with a regular pattern~\cite{Multiple-Coverage2011, k-Coverage2012}. Regular deployments need less sensor nodes than randomized ones, but they require centralized coordination and a FoI with regular-shape. A common practice is a high density of sensor nodes randomly distributed  over the monitored FoI.

In this paper, we focus on the mobile charger scheduling scenario in which the $k$-coverage ability of a network system needs to be maintained. A node sends a charging request with its position information and a charging deadline estimated based on its current residual energy and battery consumption rate.  A mobile charger seeks a path to charge sensor nodes before their charging deadlines under the constraint of maintaining the $k$-coverage ability of the monitored area, with an objective of maximizing the energy usage efficiency, i.e., minimizing the energy consumption on traveling per tour. 

We formulate the optimal $k$-coverage charging problem and show its hardness. We then construct a directed graph to model the problem and prove that it is a directed acyclic graph (DAG). We first apply dynamic programming to search for an optimal charging path. However, the computational complexity is still prohibitive for large size networks. We then introduce Deep Q-learning, a reinforcement learning algorithm to tackle the problem. Traditional heuristic or Mixed-Integer and Constraint Programming approaches to such combinatorial optimization problems need to identify domain-specific heuristics, while reinforcement learning algorithms could discover domain-specific control policies automatically. Specifically, Deep Q-Learning applies Deep Neural Network (DNN) to approximate the reward function that provides feedback to the control policy. It can handle complicated problems with huge size of states like the optimal $k$-coverage charging problem. It is also fast to learn the reward function. We implement three other heuristic algorithms for comparison and then conduct extensive simulations with experimental data included. Results demonstrate that the proposed Deep Q-Learning algorithm consistently produces optimal solutions or closed ones and significantly outperforms other methods.

The main contributions of this work are as follows:

\begin{itemize}{\setlength{\leftmargin}{0.5em}\labelwidth=0.5em
		\itemsep=0.0em\topsep=0.2em\parsep=0.1em}
	\item  We formulate the optimal $k$-coverage charging problem.
	\item We prove the NP-hardness of the optimal $k$-coverage charging problem.
	\item  We model the optimal $k$-coverage charging problem and apply dynamic programming and reinforcement learning techniques to design algorithms to tackle the  problem with extensive simulations conducted to verify their effectiveness.
\end{itemize}

The rest of the paper is organized as follows. We review the closely related works in Section~\ref{Sec:Related}. We formulate the optimal $k$-coverage charging problem and analyze its hardness in Section~\ref{Sec:Problem}. We present our algorithms in Sections~\ref{Sec:Algorithm1} and~\ref{Sec:Algorithm2}, respectively. Section~\ref{Sec:Simulation} presents the simulation results. Section~\ref{Sec:Conclusion} concludes the paper.

\section{Related Works}
\label{Sec:Related}

The problem we study in the paper is closely related with the mobile charger scheduling and traveling salesman problem with deadline. We give a brief review of the related works.

\subsection{Mobile Charger Scheduling}

The mobile charger scheduling problem optimizes the trajectory of a mobile charger (e.g., a mobile robot) to maintain the operation of a network system. There are many research works in this area with variants of the problem. Here we only list some of the most recent and representative works.

Shi et al.~\cite{shi2011renewable} consider the scenario of a wireless vehicle charger periodically traveling inside a sensor network and charging sensor nodes. They aim to minimize the time spent on path in each cycle. Multi-node wireless energy transfer technology is considered in~\cite{xie2012renewable}. The authors propose a cellular structure that partitions a two-dimensional plane into adjacent hexagonal cells such that a wireless vehicle charger visits the center of each cell and charges several sensor nodes at the same time. Xie et al.~\cite{xie2013bundling,xie2013traveling} consider the scenario of co-locating a mobile base station in a wireless charging vehicle and investigate the optimization problems of  entire system. Liang et al.~\cite{liang2017approximation} seek a charging tour that maximizes the total energy replenished to sensor nodes by a mobile charger with a limit energy capacity. Dai et al.~\cite{dai2018radiation} considers the scenario that both chargers and rechargeable devices are static. They study the optimization  problem to maximize the overall effective charging energy of all rechargeable devices and minimize the total charging time without violating the electromagnetic radiation (EMR) safety.

Energy replenishment  in robotic sensor networks is discussed in~\cite{he2014mobile,chen2016charge}. Specifically, He et al.~\cite{he2014mobile} aim to minimize the traveling distance of a mobile charger and keep all the mobile robots charged before their deadlines.  Chen et al.~\cite{chen2016charge} seek a charging path maximizing the number of mobile robots charged within a limited time or energy budget.

Multiple mobile chargers in a wireless rechargeable network raise new challenges but can work more efficiently.
Collaborative mobile charging, where mobile chargers are allowed to intentionally transfer energy between themselves, is proposed in~\cite{Collaborative2015}  to optimize energy usage effectiveness. Liang et al.~\cite{Liang:2016:MLR:2925994.2898357}  minimize the number of mobile charging vehicles to charge sensors in a large-scale WSN so that none of the sensors will run out of energy. Lin et al.~\cite{Lin18-INFOCOM} propose a real-time temporal and spatial-collaborative charging scheme for multiple mobile chargers by combining temporal requirements as well as spatial features into a single priority metric to sort real-time charging requests.

\subsection{Traveling Salesman Problem with Deadline}

The Traveling Salesman Problem (TSP), a class of combinatorial optimization problems, has been extensively studied with many approximation and heuristic algorithms proposed~\cite{cormen2009introduction}. Deadline-TSP is a relevant extensions of TSP, which contains two type of problems. One is prize collecting TSP problem with deadlines\cite{bansal2004approximation,chekuri2005recursive}, where each node has a prize and need to be visited before their deadline, and we want to maximize the prizes with the limit of time or length of the walk. $\log N$-approximation solution exists for prize collecting TSP with deadlines problem. The second type of the problem is that all the nodes need to be visited before their deadlines with a minimum traveling cost. The added time constraint seems restrict the search space of solution, but it actually renders the problem even more difficult. Even finding a feasible path for such problems is NP-complete \cite{savelsbergh1985local}. Some exact algorithms for type II TSP with deadline problem have been proposed in \cite{savelsbergh1985local,langevin1993two,dumas1995optimal}. Compared with first type of TSP-deadline problem, even finding a feasible
path for the trivial case of the type II TSP-deadline problem is an NP-complete one.

\section{Network Model and Problem Formulation}
\label{Sec:Problem}

Before giving a formal definition of the $k$-coverage charging problem studied in the paper, we first introduce the wireless sensor network model employed in this research.

\subsection{Model of Wireless Sensor Network}

We assume a set of stationary sensor nodes,  $V = \{v_i | 1 \le i \le n\}$, deployed over a planar FoI with locations, $P = \{p_i | 1 \le i \le n\}$. For each sensor node $v_i$, we assume a disk sensing model with sensing range $r$. Specifically, denote $A$ the FoI, if the Euclidean distance between a point $q \in A $ and node position $p_i$ is within distance $r$, i.e., $||p_i - q||_{L_2} \le r$, then the point $q$ is covered by the sensor $v_i$, and we use $v_i(p) = 1$ to represent it, as shown in equation \eqref{eq:equation1}:
\begin{equation}
\label{eq:equation1}
v_i(q)=
\begin{cases}
1 & ||p_i - q||_{L_2} \le r \\
0 & \text{otherwise} \\
\end{cases}
\end{equation}

\begin{definition}[Full Coverage]
	If for any point $q \in A$, there exists at least one sensor node covering it, i.e., $\sum_i v_i(q) \ge 1$, then area $A$ is full covered.
\end{definition}

\begin{definition}[$K$-Coverage]
	If for any point $q \in A$, there exist at least $k \ge 1$ sensor nodes covering it, i.e., $\sum_i v_i(q) \ge k$, then area $A$ is $k$-covered.
\end{definition}

It is obvious that full coverage is a special case of $k$-coverage with $k=1$.

\subsection{Problem Statement}
 A sensor node $v_i$ is equipped with a rechargeable battery with capacity $B$. $B_i(t)$ denotes  the residual energy of sensor node $i$ at time $t$. Charger sends its departure time denoted as $t_0$ from service station to each sensor. When receiving the message, $v_i$ estimates its residual energy at $t_0$ denoted as $B_i(t_0)$. If it is less than an energy threshold $\alpha$, i.e., $B_i(t_0)/B\le\alpha$, $v_i$ sends a charging request $(id,p_i,D_i)$ to charger.
The request includes  the ID, position, and energy exhausted time of sensor $v_i$, denoted as $id$, $p_i$, and $D_i$, respectively. The estimated energy exhausted time, i.e., charging deadline, is estimated based on the residual energy at $t_0$ and an average battery consumption rate denoted as $\beta_i$. Specifically, $D_i = B_i(t_0)/\beta_i$. Note that nodes may have different energy consumption rates.

A mobile charger with an average moving speed $s$ is responsible for selecting and charging sensors sending requests. We assume that the time spent on charging path is less then the operation time of sensors, so a sensor node only needs to be charged once in each tour. Unless under an extremely dense sensor deployment, we consider that a charger charges sensor nodes one by one because the energy efficiency reduces dramatically with distance, e.g., the energy efficiency drops to 45 percent when the charging distance is $2$m ($6.56$ ft)\cite{kurs2007power}. The energy transfer rate of charger is denoted as $r_c$.

We consider a charging path scheduling and optimization problem. A mobile charger selects and charges a number of sensor nodes before their deadlines to guarantee  $k$-coverage of area $A$, and it seeks a path with a minimum energy consumption on traveling. Specifically, the charging time is defined as the following:

\begin{definition}[Charging Time] \label{def:charging}
	Denote $P$ a charging path and $t_P(v_i)$ the charging time along $P$ at node $v_i$. If $P$ goes from nodes $v_i$ to $v_j$, the charging time begins at node $v_j$ is

	\begin{equation}
	\label{eq:equation5}
	t_P(v_j) =
	\begin{cases}
	t_P(v_i) +\frac{B-B_i(t_P(v_i))}{r_c} + \frac{d_{ij}}{s} \\
    \quad \quad \quad \text{if } t_P(v_i) + \frac{B-B_i(t_P(v_i))}{r_c} + \frac{d_{ij}}{s} \le D_j \\
	\text{inf} \\
    \quad \quad \quad \text{otherwise}, \\
	\end{cases}
	\end{equation}

	where  $d_{ij}$ the Euclidean distance between nodes $v_i$ and $v_j$ and  $s$ is the average speed of a charger. The residual energy $B_i(t)$ is estimated as $B_i(t) = B_i(t_0) - \beta_i*(t-t_0)$.
\end{definition}

\subsection{Problem Formulation}
The optimal $k$-coverage charging problem can be formulated as follows.

\begin{definition} [Optimal $k$-coverage charging problem]
Given a set of sensor nodes $V = \{v_i | 1 \le i \le n\}$, randomly deployed over a planar region A with locations $P = \{p_i | 1 \le i \le n\}$ such that every point of $A$ has been at least $k$ covered initially, the optimal $k$-coverage charging problem is to schedule a charging path $P$

\begin{subequations}\label{eq:equation2}
	\begin{align}
 &\min~ |P| \\
	&{\rm s.t.}~~ \sum_{i=1}^n v_i(q) \ge k, \forall q \in A.\label{eq:equation2b}\\
	&\quad\quad t_P(v_i)\leq D_i, \forall v_i \in P.
	\end{align}
\end{subequations}

\end{definition}
Note that a charger does not need to respond all the nodes sending requests.

\subsection{Problem Hardness}
\label{Sec:hardness}

We prove the NP-hardness of optimal $k$-coverage charging problem below.
\begin{theorem}{}
The optimal $k$-coverage charging problem is NP-hard.
\end{theorem}

\begin{proof}To prove the NP-hardness of the optimal $k$-coverage charging problem, we prove that the NP-hard problem: type II Traveling Salesman Problem with Deadline can be reduced to the trivial case of the optimal $k$-coverage charging problem in polynomial time. This type of Traveling Salesman Problem is that all the nodes need to be visited before their deadlines with a minimum traveling cost.
	
We consider a trivial case of the optimal $k$-coverage charging problem: we require $k=1$ and assume that the initial deployment of sensor nodes has no coverage redundancy. A charger needs to charge all the sensor nodes sending requests before their deadlines. It is straightforward to see that the solution of type II of Traveling Salesman Problem with Deadline is also the solution of the trivial case of the optimal $k$-coverage charging problem and vice versa. Since even finding a feasible path for this type of Traveling Salesman Problem with Deadline is NP-complete \cite{savelsbergh1985local},  the optimal $k$-coverage charging problem is NP-hard.

\end{proof}

\section{Problem Discretization}
\label{Sec:Discretization}

Given a sensor network with $n$ sensor nodes randomly deployed over a FoI, we assume the network with a reasonable density such that the area is at least $k$-coveraged initially. 

\subsection{Area Segmentation}
\label{sec:area_segmentation}

The sensing range of a sensor node $v_i$ is a disk-shape region centered at $p_i$ with radius $r_i$. These disk-shape sensing regions of a network divide a planar FoI $A$ into a set of subregions, marked as $A = \{a_i | 1 \le i \le m\}$. Then $\sum_{i=1}^n v_i(a_i)$ is the number of sensors with $a_i$ within their sensing ranges. As the definition of the optimal $k$-coverage charging problem, we assume $\sum_{i=1}^n v_i(a_i) \ge k$ in the initial deployment of a network. Denote $r(a_i)$  the number of sensor nodes sending charging requests among the $\sum_{i=1}^n v_i(a_i)$ ones. Three cases exist for subregion $a_i$:

\textbf{Case I: } $\sum_{i=1}^n v_i(a_i) - r(a_i) \ge k$: a charger may ignore all the requests.

\textbf{Case II: } $\sum_{i=1}^n v_i(a_i)= k$ and $\sum_{i=1}^n v_i(a_i) - r(a_i) < k$: a charger needs to charge all the sensor nodes sending requests with sensing ranges containing area $a_i$  within their charging windows.

\textbf{Case III: } $\sum_{i=1}^n v_i(a_i)> k$ and $\sum_{i=1}^n v_i(a_i) - r(a_i) < k$: a charger may choose to charge only $ k - \sum_{i=1}^n v_i(a_i) + r(a_i) $ sensor nodes within their charging windows among those sending requests with sensing ranges containing area $a_i$ .

A table denoted as $T$ with size $m$, is constructed to store the minimum number of sensors to charge for each $a_i$. Specifically, $T[i] = k - \sum_{i=1}^n v_i(a_i) + r(a_i)$. If the value is negative, we simply set $T[i]$ to zero.

\subsection{Time Discretization and Graph Construction} \label{sec:General_Graph_Construction}

For a sensor node $v_i$ with a charging request sent out, we divide its time window $[t_0, t_0+D_i]$ into a set of time units $\{t_i^k | 0 \le k \le D_i\}$, where $t_i^0 = t_0$ and $t_i^{D_i} = D_i$. We represent node $v_i$ with a set of discretized nodes $\{v_i(t_i^k) | 0 \le k \le D_i\}$, where $v_i(t_i^k)$ represents Node $v_i$ at time $t_i^k$. 

We then construct a directed graph denoted as $G$ with vertices and edges defined as follows.

\textbf{Vertices.} The vertex set $V(G)$ includes the discretized sensor nodes sending charging requests, i.e., $\{v_i(t_i^k) |  0 \le k \le D_i\}$.

\textbf{Edges.}
There exists a directed edge $\overrightarrow{v_i(t_i^{k})v_j(t_j^{k'})}$ from $v_i(t_i^{k})$ to $v_j(t_j^{k'})$ in the edge set $E(G)$ if and only if
\[
t_i^{k}+\frac{B-B_i}{r_c} + \frac{d_{ij}}{s} > t_j^{k'-1}, 
\]
\[
t_i^{k}+\frac{B-B_i}{r_c}+\frac{d_{ij}}{s}\le t_j^{k'},
\]
where $k' > 0$.

A directed edge ensures that a charger arrives at sensor node $v_j$ before its deadline. The charging time is the arrival time of the charger.

\begin{theorem}{}
	$G$ is a directed acyclic graph (DAG).
\end{theorem}
\begin{proof}
	Suppose there exists a cycle in $G$. Assume vertices $v_i(t_i^k)$ and $v_j(t_j^{k'})$ are on the cycle. Along the directed path from $v_i(t_i^k)$ to $v_j(t_j^{k'})$, it is obvious that $t_i^k < t_j^{k'}$. However, along the directed path from $v_j(t_j^{k'})$ to $v_i(t_i^k)$, we have $t_j^{k'} < t_i^k$. Contradiction, so $G$ is a directed acyclic graph.
\end{proof}

\begin{definition}[Clique]
	A set of nodes $\{v_i(t_i^k) | 0 \le k \le D_i\}$ in $G$ is defined as a clique if they correspond to the same  node $v_i$ at different time units.
\end{definition}

\begin{definition}[Feasible Path]
	A path $P$ in $G$ is a feasible one if it passes no more than one vertex of a clique. At the same time, charging along $P$ satisfies the $k$-coverage requirement of the given network.
\end{definition}

\section{Dynamic Programming Algorithm}
\label{Sec:Algorithm1}

Considering that the constructed graph $G$ is a DAG that can be topologically sorted, we design a dynamic programming algorithm to find an optimal charging path for the $k$-coverage charging problem.

To make sure that the computed charging path passes no more than one discretized vertex of a sensor node, we apply the color coding technique introduced in~\cite{alon1995color} to assign each vertex a color. Specifically, we generate a coloring function $c_v: V \rightarrow \{1,...,n\}$ that assigns each sensor node a unique node color. Each sensor node then passes its node color to its discretized ones. A path in $G$ is said to be colorful if each vertex on it is colored by a distinct node color. It is obvious that a colorful path in $G$ passes no more than one discretized vertex of a sensor node.

To take into the consideration of traveling distance from service station to individual sensor node, we add an extra vertex denoted as $v_0$ and connect it with directed edges to vertices in $G$, i.e., $\{ v_i(t_i^k) | k == 0 \}$. The length of edge $\overrightarrow{v_0v_i(t_i^0)}$ is the Euclidean distance between the service station and sensor node $v_i$. The table $T$ constructed in Sec.~\ref{sec:area_segmentation} is stored at $v_0$.

We first topologically sort the new graph, i.e., $G + v_0$. Then we start from $v_0$ to find colorful paths by traversing from left to right in linearized order. Specifically,  $v_0$ checks neighbors connected with outgoing edges and sends table $T$ to those contributing to the decrease of at least one table entry. Once a vertex $v_i(t_i^0)$ receives $T$, $v_i(t_i^0)$ checks the subregions within its sensing range and updates the corresponding entries of $T$. $v_i(t_i^0)$ also generates a color set $C=\{ c(v_i(t_i^0)) \}$ and stores with $T$, which indicates a colorful path of length $|C|$.

Similarly,  suppose the algorithm has traversed to vertex $v_i(t_i^k)$, we check each color set $C$ stored at $v_i(t_i^k)$ and its outgoing edge $v_j(t_j^{k'})$. If $c(v_j(t_j^{k'})) \not\in C$ and charging $v_j(t_j^{k'})$ helps decrease at least one entry of $T$ associated with $C$, we add the color set $C = \{ C + c(v_j(t_j^{k'}))\}$ along with the updated $T$ to the collection of $v_j(t_j^{k'})$.

After the update of the last vertex  in linearized order, we check the stored $T$s in each node and identify those with all zero entries. A color set associated with a $T$ with all zero entries represents a colorful path that is a feasible solution of the $k$-coverage problem.

A path can be easily recovered  from a color set. The basic idea is to start from vertex $v_i(t_i^k)$ with a color set $C$. We check the stored color sets of vertices connected to $v_i(t_i^k)$ with incoming edges. Assume we identify a neighbor node $v_j(t_j^{k'})$ storing a color set $C-c(v_i(t_i^k))$, then we continue to trace back the path from $v_j(t_j^{k'})$ with a color set $C-c(v_i(t_i^k))$. When we trace back to $v_0$, we have recovered the charging path. Among all feasible charging paths, the one with a minimal traveling distance is the optimal one.

\begin{lemma}{}
	The algorithm returns an optimal solution of the $k$-coverage charging problem, i.e., a feasible path maximizing the energy usage efficiency, if it exists.
\end{lemma}
\begin{proof}
	We first show that the algorithm returns a feasible path. A path returned by the algorithm is a colorful one that guarantees the path passes a sensor node no more than once. In the meantime, charging time at each sensor node along the path is before its deadline, otherwise a directed edge along the path won't exist. Array $T$ with all zero entries makes sure that the $k$-coverage is maintained.
	
	When the algorithm has traversed to the $i^{th}$ node in linearized order, each colorful path passing through the node has been stored in the node.
\end{proof}

Note that the computational complexity of the dynamic programming algorithm can increase exponentially in the worst case because the stored color sets at a vertex can increase exponentially to the size of sensor nodes n.

\section{Deep Q-Learning Algorithm}\label{Sec:Algorithm2}

\subsection{Motivation}

A combinatorial optimization problem searches for an optimal solution of an objective function under a set of constraints. The domain of the objective function is discrete, but prohibitive for an exhaustive search. An optimal solution is a feasible one satisfying the set of constraints and minimizing the value of the objective function.

Previous heuristic or mixed-integer and constraint programming approaches to combinatorial optimization problems need to identify domain-specific heuristics, while reinforcement learning algorithms discover domain-specific  control policies automatically by exploring different solutions and evaluating their qualities as rewards. These rewards are then provided as feedback to improve the future control policy~\cite{sutton1998reinforcement}.

The optimal $k$-coverage charging problem is a combinatorial optimization one. We introduce Deep Q-learning, a reinforcement learning algorithm to tackle the optimal $k$-coverage charging problem. ``Q" stands for the quality or reward of an action taken in a given state. In Q-learning, an agent maintains a state-action pair function stored as a Q-table $Q[S,A]$ where $S$ is a set of states and $A$ is a set of actions. A Q-value of the table estimates how good a particular action will be in a given state, or what reward the action is expected to bring.

Q-Learning is a model-free reinforcement learning algorithm, so there is no need to find all the combinations to check the existence of a solution. An agent will choose next state based on current one and stored state-action rewards, requiring less computation and storage space comapred with model-based reinforcement learning algorithms. Q-Learning is a also temporal-difference reinforcement learning algorithm. An agent  can learn online after every step, and even from an incomplete sequences (a situation that leads to unfeasible solution). In theory, Q-learning has been proven to converge to the optimal Q-function for an arbitrary target policy given sufficient training~\cite{ConvergenceQ-learning}.

Deep Q-Learning applies Deep Neural Network (DNN) to approximate the Q-function, i.e., a deep neural network that takes a state and approximates the Q-value for each state-action pair. It can handle the situation when the number of states of a problem is huge. It is also faster to learn the reward value of each state-action pair than Q-learning. Therefore, deep Q-learning can handle more complicated problems compared with Q-Learning.

An end-to-end deep Q-learning framework introduced in~\cite{khalil2017learning} automatically learns greedy heuristics for hard combinatorial optimization problems on graphs. We modify the framework in~\cite{cheng2007modified} and combine with deep Q-learning in~\cite{mnih2013playing} to tackle the optimal $k$-coverage charging problem.

\subsection{Graph Construction}\label{General_Graph_Construction}
We first construct a directed graph denoted as $G$ as input for the Deep Q-Learning algorithm. Notice that the reason that we reconstruct graph is that sensor nodes do not need to be discretized by using DQN because agent will know the exact time when it comes to the node and it could decide to go ahead or abandon the path after checking the current residual energy $B_i(t)$ and the deadline constraint. This graph construction can reduce the number of action of the DQN and reduce the training and test time. The vertices and edges of $G$ are defined as follows.

\textbf{Vertices.}
The vertex set $V(G)$ includes all the sensor nodes, i.e.,  $\{v_i | 1 \le i \le n\}$ and a start vertex denoted as $v_0$. Each vertex $v_i$ has a deadline $D_i$. The location of $v_0$ is the service station and the deadline of $v_0$ is $\text{inf}$. Note that the deadline of sensor nodes without sending any request is set as $0$.

\textbf{Edges.}
For any sensor nodes $v_i$ and $v_j$ in $V(G)$, $d_{ij}$ is the euclidean distance between $v_i$ and $v_j$. There exists an edge $\overrightarrow{v_iv_j}$ in $V(G)$ if and only if the inequality below holds
\begin{equation}
\label{eq:equation4}
\frac{B-B_i(t_0)}{r_c}+\frac{d_{ij}}{s}\le D_j
\end{equation}
where $B_i(t_0)$ denotes as the residual energy of sensor node $i$ at charger departure time $t_0$, $r_c$ is the energy transfer rate of charger and $s$ is the average speed of a charger. In this way, we can make sure that all the feasiable path exist in the graph, but if the path valid or not need to be checked according the time $t$ when agent comes to the node. Because compared with Def.~\ref{def:charging}, this edge construction do not consider the $t_P(v_i)$ and only consider the initial energy $B_i(t_0)$ instead of $B_i(t_P(v_i))$ which related with the path $P$. If a charging path $P$ goes from node $v_i$ to node $v_j$, the charging time beginning at node $v_j$ is given by Def.~\ref{def:charging}.

\begin{definition}[Feasible Path]
	A path $P$ in $G$ is a feasible one if it starts from and ends at $v_0$, and has no repeated vertex and the charging time at each vertex is not inf. At the same time, charging along $P$ satisfies the $k$-coverage requirement of the given network.
\end{definition}

\subsection{Deep Q-Learning Formulation}\label{Deep Q-Learning Formulation}

 We define the states, actions, rewards, and stop function in the deep Q-learning framework as follows:

\textbf{States: }  A state $S$ is a partial solution $S\subseteq V(G)$, an ordered list of visited vertices. The first vertex in $S$ is $v_0$.

\textbf{Actions: } Let $\bar{S}$ contain vertices not in $S$ and has at least one edge from vertices in $S$. An action is a vertex $v$ from $\bar{S}$ returning the maximum reward. After taking the action $v$, the partial solution $S$ is updated as
\begin{equation}
S' := (S, v), \text{where } v = \arg\max_{v\in\bar{S}}Q(S,v)
\end{equation}
$(S,v)$ denotes appending $v$ to the best position after $v_0$ in $S$ that introduces the least traveling distance and maintains all vertices in the new list a valid charging time.

\textbf{Rewards: } The reward function $R(S,v)$ is defined as the change of the traveling distance when taking the action $v$ and transitioning from the state $S$ to a new one $S'$. Assume $v_i$, $v_j$ are two adjacent vertex in $S$, $v_0$ is the first vertex in the $S$, and $v_t$ is the last vertex in the $S$.

The reward function $R(S,v_k)$ is defined as follows:
\begin{equation}
\label{eq:equation7}
R(S,v_k) =
\begin{cases}
-\min({d_{ik}}+{d_{kj}}-{d_{ij}},{d_{tk}}+{d_{k0}}- {d_{t0}}) & t_{S'}(v) \neq \text{inf}\\
-\text{inf} & \text{otherwise} \\
\end{cases}
\end{equation}
where $d_{ij}$ is the euclidean distance between nodes $v_i$ and $v_j$. $t_{S'}(v)$ stands for the new charging time of each node in path formed by $S'$ after inserting the $v_k$. For example, $t_{S'}(v_k)$, $t_{S'}(v_t)$ are the charging time of $v_k$, $v_t$ by the Def.~\ref{def:charging} along the charging path formed by $S'$ after inserting the $v_k$ to $S$, respectively.

\textbf{Stop function: }: The algorithm terminates when the current solution satisfies the $k$-coverage requirement or  R(S,v) is -inf.

\subsection{Deep Q-Learning Algorithm}
Algorithm~\ref{DQN_algorithm} summarizes the major steps of our algorithm. Briefly, $Q(S,v;\Theta)$ is parameterized by a deep network with a parameter $\Theta$ and learned by one-step Q-learning. At each step of an episode, $\Theta$ is updated by a stochastic gradient descent to minimize the squared loss:

\[
(y_l - Q(S_l,v_l;\Theta))^2
\]
where
\begin{equation}\label{eq:q_learning}
y_l =
\begin{cases}
R(S_l,v_l)+\gamma\max_{v'} Q(S_{l+1},v';\Theta)\\
\quad \quad \quad \text{if } S_{l+1} \text{ non-terminal and } v'\in \bar{S_l} \\
R(S_l,v_l) \\
\quad \quad \quad \text{otherwise}
\end{cases}
\end{equation}

We use experience replay method in~\cite{mnih2013playing} to update $\Theta$, where the agent's experience in each step is stored into a dataset. When we update $\Theta$, we sample random batch from dataset that is populated from previous episodes. The benefits include increasing data efficiency, reducing correlations between samples, and avoiding oscillations with the parameters.
\begin{algorithm}
	\caption{Deep Q-learning Algorithm}
	\begin{algorithmic}[1]
		\STATE Initialize replay memory $\mathcal{H}$ to capacity $C$
		\FOR{ each episode }
		\STATE Initialize state $S_1 = (v_0)$
		\FOR{ step $m = 1$ \TO $M$ }
		\STATE Select  $v_m = \arg\max_{v\in\bar{S}_m}Q(S_m,v;\Theta)$ with probability $1-\epsilon$
		\STATE Otherwise select a random vertex $v_m\in\bar{S}_m$
		\STATE Add $v_m$ to partial solution: $S_{m+1}:=(S_m,v_m)$
		\STATE Calculate reward $R(S_m,v_m)$ by \eqref{eq:equation7}
		\STATE Store tuple $(S_{m},v_{m},R(S_m,v_m),S_{m+1})$ to $\mathcal{H}$
		\STATE Sample random batch $(S_l,v_l,R(S_l,v_l),S_{l+1})$ from $\mathcal{H}$
		\STATE Update the network parameter $\Theta$ by \eqref{eq:q_learning}
		\IF{$S_{m+1}$ satisfy the stop function}
		\STATE Break
		\ENDIF
		\ENDFOR
		\ENDFOR		
	\end{algorithmic}
	\label{DQN_algorithm}
\end{algorithm}

\section{Performance Evaluation}\label{Sec:Simulation}
\subsection{Simulation Settings}

We set up an Euclidean square $[500, 500]m^2$ as a simulation area and randomly deploy sensor nodes ranging from $32$ to $80$ in the square such that the area is at least $k$-coverage initially where $k$ varies from $2$ to $4$. The sensing range $r$ is $135$m. The base station and service station of charger are co-located in the center of the square. A charger with a starting point from the service station has an average traveling speed $5$m/s and consumes energy $600J/m$ \cite{liang2017approximation}. The battery capacity $B$ of each sensor is $10.8$KJ \cite{liang2017approximation}. The remaining energy threshold $\alpha$ vary from $0.2$ to $0.8$. A sensor sends a charging request before the leaving of the charger from the service station. The sensor will include in the request the estimated energy exhausted time based on its current residual energy and energy consumption rate. To simulate such request, we consider the residual battery of a sensor is a uniform random variable $B_i$ between $(0.54,10.8]$KJ \cite{shi2011renewable}. The energy exhausted time $D_i = B_i/\beta_i$. We choose the energy consumption rate $\beta_i$ from the historical record of real sensors in~\cite{Zhu:2009:LES:1555816.1555849} where the rate is varying according to the remaining energy and arrival charging time to the sensor. The energy transfer rate $r_c$ is $20$W \cite{chen2016charge}. The discredited time stepsize is $1$s.

We add three heuristic algorithms including Ant Colony System (ACS) based algorithm, Random algorithm, and Greedy algorithm for comparison. Sec.~\ref{Sec:Comparison_Algorithms} explains the implementations of the three algorithms in detail.

Traveling energy and computing time are important metrics to evaluate the performances of different algorithms. Network settings may also affect the performance. Therefore, we study the impact of the three parameters of network setting: the coverage requirement $k$, the size of sensor network $n$, and the remaining energy threshold $\alpha$ in Secs.~\ref{Sec:K},~\ref{Sec:Size}, and~\ref{Sec:Percentage}, respectively.

\subsection{Comparison Algorithms}\label{Sec:Comparison_Algorithms}

We implement three  heuristic algorithms for comparison: Ant Colony System (ACS) based algorithm, Random algorithm, and Greedy algorithm.

ACS algorithm solves the traveling salesmen problem with an approach similar to the foraging behavior of real ants~\cite{ant1991,dorigo1997ant,Gutjahr:2000:GAS:348599.348602}. Ants seek path from their nest to food sources and leave a chemical substance called pheromone along the paths they traverse. Later ants sense the pheromone left by earlier ones and tend to follow a trail with a stronger pheromone. Over a period of time, the shorter paths between the nest and food sources are likely to be traveled more often than the longer ones. Therefore, shorter paths accumulate more pheromone, reinforcing these paths.

Similarly, ACS algorithm places agents at some vertices of a graph. Each agent performs a series of random moves from current vertex to a neighboring one based on the transition probability of the connecting edge. After an agent has finished its tour, the length of tour is calculated and the local pheromone amounts of edges along the tour are updated based on the quality of the tour. After all agents have finished their tours, the shortest one is chosen and the global pheromone amounts of edges along the tour are updated. The procedure continues until certain criteria are satisfied. When applying ACS algorithm to solve the the traveling salesmen problem with deadline, two local heuristic functions are introduced in~\cite{cheng2007modified} to exclude paths that violate the deadline constraints.

We modify ACS algorithm introduced in~\cite{cheng2007modified} for the optimal $k$-coverage charging problem. Agents start from and end at $v_0$.  Denote $\tau_{ij}(t)$ the amount of global pheromone deposited on edge $\overrightarrow{v_iv_j}$ and $\Delta\tau_{ij}(t)$ the increased  amount at the $t^{th}$ iteration. $\Delta\tau_{ij}(t)$ is defined as
\begin{equation}
\label{eq:Deltatau_ij}
\Delta\tau_{ij}(t) =
\begin{cases}
\frac{1}{L^*}  & \text{if } \overrightarrow{v_iv_j} \in P^*\\
0 & \text{otherwise} \\
\end{cases}
\end{equation}
where $L^*$ is the traveling distance of the shortest feasible tour $P^*$ at the $t^{th}$ iteration. Global pheromone $\tau_{ij}(t)$ is updated according to the following equation:
\begin{equation}
\label{eq:tau_ij}
\tau_{ij}(t) = (1-\theta)\tau_{ij}(t-1) + \theta\Delta\tau_{ij}(t),
\end{equation}
where $\theta$ is the global pheromone decay parameter. The local pheromone is updated in a similar way, where $\theta$ is replaced by a local pheromone decay parameter and $\Delta\tau_{ij}(t)$ is set as the initial pheromone value.

We also modify the stop criteria of one agent such that the traveling path satisfies the requirement of $k$-coverage, or the traveling time of current path is inf, or the agent is stuck at a vertex based on the transition rule.

Random and Greedy algorithms work much more straightforward. The Random algorithm randomly chooses a next node not in the same clique of the existing path and with an outgoing edge from current one to charge. The Greedy algorithm always chooses the nearest node not in the same clique of the existing path with an outgoing edge from current one. Random and Greedy algorithms terminate when they either find a feasible path or are locally stuck. Note that for ACS and Random algorithms, we always run multiple times and choose the best solution.

\begin{table}
	\caption{Performance comparison under different coverage requirement $k$\label{tab:parameter_k_result}}
	\begin{center}
		\ra{1.3}
		{    \fontsize{9pt}{9pt}\selectfont
			\begin{tabular}[t]{ @{}  c | c | c | c |c  @{}} \toprule
	&   & \multicolumn{3}{c}{$n=64$, $\alpha = 0.45$} \\ 	\cline{3-5}	
	Algorithm & $k$  & Computation& Feasible & Traveling \\
	&   & Time  &  Path &  Energy \\
	&   & (s) & Found  &   (kJ) \\ \hline			
				Dynamic  & \multirow{5}{*}{2} & 0.102 & Yes    & 249 \\
				DQN  &  & 16 & Yes    & 249 \\
				ACS  & &  5  & Yes    & 249 \\
				Random & & 0.0006 & Yes    & 249  \\
				Greedy & & 0.0007 & Yes    & 288 \\ \hline
				Dynamic  & \multirow{5}{*}{3} & 455 & Yes    & 702  \\
				DQN   &  & 20 & Yes    & 702  \\
				ACS   &  & 19 & Yes    & 702  \\
				Random & & 0.0003 & No     & --   \\
				Greedy & & 0.0003 & Yes     & 846   \\ \hline
				Dynamic   & \multirow{5}{*}{4} & -- & --    & -- \\
				DQN   &  & 71 & Yes    & 1089 \\
				ACS   &  & 73 & Yes     & 1188  \\
				Random & & 0.0006 & No     & --  \\
				Greedy & & 0.0005 & Yes     & 1254  \\
				\bottomrule
			\end{tabular}
		}
	\end{center}
\end{table}

\subsection{Coverage Requirement} \label{Sec:K}

We set the number of sensor nodes $n= 64$ and the remaining energy threshold $\alpha = 0.45$. The coverage requirement $k$ varies from $2$ to $4$. Table~\ref{tab:parameter_k_result} gives the performances of different algorithms. The higher the coverage requirement $k$ is, the more the sensor nodes need to be charged. The trend is obvious in Table~\ref{tab:parameter_k_result} that the traveling energy increases with the increase of the coverage requirement $k$.  The Random algorithm can only find a feasible path when $k$ is small. The Greedy algorithm can find a feasible solution for different $k$, but with traveling energy much higher than others.
The ACS algorithm can find a feasible path for different $k$ too, with traveling energy less than the Random and Greedy ones. The DQN algorithm can find the optimal path for all $k$s. At the same time, the computing time of DQN including both the training and computing time, grows slowly with the increase of $k$. By contrast, the computing time of the dynamic programming algorithm increases exponentially. The dynamic programming algorithm runs out of the memory when $k$ is large. Overall, the performance of DQN significantly outperforms all other comparison algorithms. It is worth mentioning that the Random algorithm performs better than the Greedy one in some case because the nearest node may not be a good choice and a randomly chosen one may lead the searching of a feasible path out of stuck.

\begin{table}
	\caption{Performance comparison under different sizes of sensor network $n$ \label{tab:parameter_n_result}}
	\begin{center}
		\ra{1.3}
		{    \fontsize{9pt}{9pt}\selectfont
			\begin{tabular}[t]{ @{}  c | c | c | c |c  @{}} \toprule
	& & \multicolumn{3}{c }{$k=3$, $\alpha =0.45$}  \\ 	\cline{3-5}
	Algorithm & $n$  & Computation & Feasible  & Traveling \\
	&   & Time  &  Path &  Energy \\
	&   & (s) & Found  &   (kJ)  \\ \hline			
				Dynamic  & \multirow{5}{*}{48} & 156700 & Yes & 771 \\
				DQN  &&   83 & Yes & 771   \\
				ACS  &&  92 & Yes & 951  \\
				Random & & 0.0004 & Yes & 2085\\
				Greedy & & 0.0004 & Yes & 888\\ \hline
				Dynamic  & \multirow{5}{*}{64} & 455 & Yes    & 702  \\
				DQN   &  & 20 & Yes    & 702  \\
				ACS   &  & 19 & Yes    & 702  \\
				Random & & 0.0003 & No     & --   \\
				Greedy & & 0.0003 & Yes     & 846   \\ \hline
				Dynamic   & \multirow{5}{*}{72} & 362 & Yes & 567 \\
				DQN   & & 32 & Yes & 567\\
				ACS   & & 57 & Yes & 567 \\
				Random & & 0.0003 & Yes & 1941 \\
				Greedy & & 0.0003 & Yes & 810 \\ \hline
                Dynamic   & \multirow{5}{*}{80} & 268 & Yes & 345 \\
				DQN   & & 20 & Yes & 345\\
				ACS   & & 18 & Yes & 345\\
				Random & & 0.0003 & No & -- \\
				Greedy & & 0.0003 & Yes & 375\\
				\bottomrule
			\end{tabular}
		}
	\end{center}
\end{table}

\subsection{Size of Sensor Network}\label{Sec:Size}

We set the coverage requirement $k=3$ and the  remaining energy threshold $\alpha = 0.45$. The size of a sensor network $n$ varies from $48$ to $80$. Table~\ref{tab:parameter_n_result} gives the performances of different algorithms. It is obvious that the traveling energy in Table~\ref{tab:parameter_n_result}  decreases with the increase of $n$ because there are more redundant sensor nodes to maintain the $k$-coverage of a network. The Random algorithm fails to detect a feasible path and the result of the Greedy algorithm is far from the optimal one. The ACS algorithm performs better when $n$ is large and the network has more redundant sensor nodes to provide coverage. Again, the DQN algorithm still performs the best,  finding the optimal paths for all $n$s with a stable computing time. 

\subsection{Remaining Energy Threshold}\label{Sec:Percentage}

Tables~\ref{tab:parameter_alpha_result_1} and~\ref{tab:parameter_alpha_result_2} give the performances of different algorithms with the remaining energy threshold $\alpha$ varying from $0.2$ to $0.8$ under two network settings: $n=32$ and $k=2$, and $n=48$ and $k=3$, respectively. With the increased remaining energy threshold, the traveling energy increases in both network settings. The Random and Greedy algorithms fail to detect a feasible path in many cases. The dynamic programming algorithm runs out of the memory when $\alpha$ is large. Both the ACS and DQN algorithms find feasible paths for all cases. However, the performance of the ACS algorithm decreases with the increase of $\alpha$. The DQN algorithm consistently and significantly outperforms all other comparison algorithms including botht the traveling energy and computing time.

\begin{table}
	\caption{Performance comparison under different remaining energy threshold $\alpha$ when $k=2$, $n = 32$\label{tab:parameter_alpha_result_1}}
	\begin{center}
		\ra{1.3}
		{    \fontsize{9pt}{9pt}\selectfont
			\begin{tabular}[t]{ @{}  c | c | c | c | c @{}} \toprule
				 &   & \multicolumn{3}{c}{$k=2$, $n = 32$} \\ 	\cline{3-5}			
				Algorithm & $\alpha$  & Computation & Feasible  & Traveling  \\
&   & Time  &  Path &  Energy \\
&   & (s) & Found  &   (kJ)  \\ \hline	
				 Dynamic  & \multirow{5}{*}{0.2} & 0.0012 & Yes    & 405\\
				DQN    && 13&  Yes    & 405\\
				ACS    && 3&  Yes    & 405\\
				Random && 0.0002&  No    & --\\
				Greedy && 0.0003&  Yes     & 420\\ \hline
				
				Dynamic  & \multirow{5}{*}{0.4} & 9760&  Yes    & 696\\
				DQN    && 26&  Yes    & 696  \\
				ACS    & & 38 &  Yes    & 855 \\
				Random && 0.0003 & No     & --   \\
				Greedy & & 0.0003&  No     & 891 \\ \hline
				
				Dynamic  &  \multirow{5}{*}{0.6} & 135742 & Yes    & 1071  \\
			DQN    && 116	&  Yes    & 1071  \\
				ACS    & & 232&  Yes    & 2289  \\
				Random & & 0.0006&  No     & --   \\
				Greedy && 0.0004 &  No     & --   \\ \hline
				
			Dynamic  & \multirow{5}{*}{0.8} & -- &  --        & -- \\
				DQN    && 136 & Yes        & 1080 \\
				ACS    && 400 &  Yes         & 2544    \\
				Random && 0.002 &  No         & --   \\
				Greedy & & 0.001& No        & --   \\
				\bottomrule
			\end{tabular}
		}
	\end{center}
\end{table}
\begin{table}
	\caption{Performance comparison under different remaining energy threshold $\alpha$ when $k=3$, $n = 48$\label{tab:parameter_alpha_result_2}}
	\begin{center}
		\ra{1.3}
		{   \fontsize{9pt}{9pt}\selectfont
			\begin{tabular}[t]{ @{}  c | c | c | c | c @{}} \toprule
				 &   & \multicolumn{3}{c }{$k=3$, $n = 48$}  \\ 	\cline{3-5}			
				Algorithm & $\alpha$  & Computation & Feasible  & Traveling \\
&   & Time  &  Path &  Energy \\
&   & (s) & Found  &   (kJ) \\ \hline	
				 Dynamic  & \multirow{5}{*}{0.2} & 0.02 & Yes    & 348 \\
				DQN    && 10 & Yes    & 348 \\
				ACS    && 1 & Yes    & 348 \\
				Random && 0.0003 & No    & -- \\
				Greedy && 0.0002 & Yes    & 348 \\ \hline
				
				Dynamic  & \multirow{5}{*}{0.4} & 145602& Yes & 750 \\
				DQN  &&   81 & Yes & 750   \\
				ACS  &&  90 & Yes & 930  \\
				Random & & 0.0004 & Yes & 2070\\
				Greedy & & 0.0004 & Yes & 870\\ \hline
				
				Dynamic  &  \multirow{5}{*}{0.6} & -- & --    & --\\
			    DQN    && 138 & Yes    & 1020\\
				ACS    && 466 & Yes    & 1521   \\
				Random & & 0.001 & No    & --   \\
				Greedy && 0.001 & No    & --   \\ \hline
				
			    Dynamic  & \multirow{5}{*}{0.8} & --  & --    & --   \\
				DQN    && 481 & Yes    & 1272   \\
				ACS    && 4200   & Yes    & 4788   \\
				Random && 0.01 & No    & --   \\
				Greedy && 0.02 & No    & --   \\
				\bottomrule
			\end{tabular}
		}
	\end{center}
\end{table}

\section{Conclusions}\label{Sec:Conclusion}
We explore the mobile charger scheduling and path optimization problem that the $k$-coverage ability of a wireless rechargeable sensor network system needs to be maintained. We formulate the problem and show its hardness. We model the problem and apply both dynamic programming and reinforcement learning techniques  to design algorithms to tackle the optimal $k$-coverage charging problem. We also implement three other heuristic algorithms for comparison. Extensive simulations with experimental data included demonstrate that the proposed Deep Q-Learning algorithm consistently produces optimal solutions or closed ones and significantly outperforms other methods.

\bibliographystyle{ieeetr}
\bibliography{bio}
\end{document}